\documentclass[prd,nofootinbib,preprintnumbers
,preprint
]{revtex4}
\usepackage{empheq}

\usepackage{amsmath}
\usepackage{amsfonts}
\usepackage{amssymb}
\usepackage{dsfont}
\usepackage{bm}
\usepackage[hyperfootnotes=false]{hyperref}
\usepackage[dvipsnames]{xcolor}
\usepackage{apptools}

\newcommand{\E}{\mathcal E}
\newcommand{\Lie}{\mathcal L}
\newcommand{\RM}{R}
\newcommand{\RS}{\mathcal R}
\newcommand{\LCM}{\nabla}
\newcommand{\LCS}{\mathcal D}
\newcommand{\SFS}{\chi}

\newcounter{example}[section]
\newcounter{remark}[section]
\newcounter{theorem}[section]
\newcounter{proposition}[section]
\newcounter{lemma}[section]
\newcounter{corollary}[section]
\newcounter{definition}[section]

\AtAppendix{\counterwithin{example}{section}}
\AtAppendix{\counterwithin{remark}{section}}
\AtAppendix{\counterwithin{theorem}{section}}
\AtAppendix{\counterwithin{proposition}{section}}
\AtAppendix{\counterwithin{lemma}{section}}
\AtAppendix{\counterwithin{corollary}{section}}
\AtAppendix{\counterwithin{definition}{section}}

\def\theremark{\arabic{section}.\arabic{remark}}
\def\thetheorem{\arabic{section}.\arabic{theorem}}

\def\thedefinition{\arabic{section}.\arabic{definition}}

\renewcommand*{\email}[1]{\footnote{Electronic address: \href{mailto:#1}{\nolinkurl{#1}} }}

\newenvironment{proof}{\noindent {\textit{Proof:}}
}{\medskip}

\newenvironment{theorem}{\refstepcounter{theorem}
\medskip\noindent{\bf Theorem \thetheorem}:\em}{\medskip}
\newenvironment{proposition}{\refstepcounter{theorem}\medskip\noindent{\bf
Proposition \thetheorem}:\em}{\medskip}

\newenvironment{definition}{\refstepcounter{definition}\medskip\noindent{\bf
Definition \thedefinition}:\em}{\medskip}

\begin{document}

\title{Uniqueness of the static vacuum asymptotically flat spacetimes with massive particle spheres}
\author{Kirill Kobialko${}^1$\email{kobyalkokv@yandex.ru}}
\author{Igor Bogush${}^2$\email{igbogush@gmail.com}}
\author{Dmitri Gal'tsov${}^1$\email{galtsov@phys.msu.ru}}
\affiliation{
${}^1$ Faculty of Physics, Moscow State University, 119899, Moscow, Russia
\\
${}^2$ Moldova State University, strada Alexei Mateevici 60, 2009, Chi\c{s}in\u{a}u, Moldova
}

\begin{abstract}
In this paper, we establish that a four-dimensional static vacuum asymptotically flat spacetime containing a massive particle sphere is isometric to the Schwarzschild spacetime. Our results expand upon existing uniqueness theorems for static vacuum asymptotically flat spacetimes, which focus on scenarios featuring event horizons or photon spheres. Similarly to the uniqueness theorems concerning photon spheres or event horizons, only a single massive particle sphere is sufficient to obtain a unique solution. However, in contrast to previous theorems, our result leads to the existence of an entire spacetime foliation sliced by a set of massive particle spheres spanning various energies.
\end{abstract}

\maketitle

\setcounter{page}{2}

\section{Introduction}

Black hole shadows offer a direct means to observe the optical characteristics of immensely strong gravitational fields. The theoretical comprehension of these shadows is closely intertwined with the photon and massive particle surfaces \cite{Perlick:2021aok,Grenzebach:2014fha,Shoom:2017ril,Cunha:2018acu,Song:2022fdg,Kobialko:2023qzo,Chen:2022scf}.

Since the inception of general relativity, it has been widely recognized that the spherically symmetric Schwarzschild solution contains a series of circular null orbits, collectively constituting a complete photon sphere due to its inherent symmetry. The profound implications of these surfaces began to crystallize in the late 1990s. The seminal work by Virbhadra and Ellis \cite{Virbhadra:1999nm} delineated the correlation between the properties of photon spheres and the intricacies of strong gravitational lensing, leading to the formal characterization of the photon sphere as a timelike hypersurface in spacetime, where the light beam's deflection angle at the closest approach distance tends towards infinity. Subsequently, Claudel, Virbhadra, and Ellis \cite{Claudel:2000yi} introduced a comprehensive definition of the general photon surface as a timelike surface wherein any null geodesic, touching it tangentially, remains entirely within it, and established a theorem linking this definition to the hypersurface's geometry. The equivalence of these definitions was demonstrated in Ref. \cite{Virbhadra:2002ju} for general static spherically symmetric metrics.

Notably, it was revealed that the intimate connection between photon spheres and strong lensing persists even in the context of naked singularities, thereby suggesting their categorization into weak and strong variants \cite{Virbhadra:2002ju,Virbhadra:2007kw}. Recently, in close correlation with photon surfaces, several significant relationships have been unveiled regarding the geometric attributes of relativistic images \cite{Virbhadra:2022iiy,Virbhadra:2024pru}, the compactness of supermassive dark objects at galactic cores \cite{Virbhadra:2022ybp}, and the impact of the cosmological constant on the photon sphere \cite{Adler:2022qtb}.

An important property of the photon surfaces is established by the theorem asserting that these are timelike totally umbilic hypersurfaces in spacetime \cite{Chen} exhibiting proportionality of their first and second fundamental forms. This purely geometric approach serves as a constructive definition for analyzing photon surfaces instead of solving geodesic equations and plays a decisive role in the analysis of the black hole uniqueness.

The first black hole uniqueness theorem was established by Israel in Ref. \cite{Israel:1967wq}. It states that the Schwarzschild solution is the only static asymptotically flat vacuum spacetime which has a nonsingular closed simply connected event horizon. Then, similar uniqueness theorem were suggested to generalize the result, e.g., the Kerr solution was shown to be the only rotating vacuum black hole by Robinson in Ref. \cite{Robinson:1975bv}. The focus on the event horizon in the uniqueness proofs was then moved to photon surfaces in work by Cederbaum \cite{Cederbaum:2014gva}, where the Schwarzschild spacetime was shown to be the only static vacuum asymptotically flat spacetime that possesses a suitably defined photon sphere. Later, the uniqueness theoremes for photon spheres were established for Einstein-scalar \cite{Yazadjiev:2015hda}, Einstein-Maxwell \cite{Yazadjiev:2015jza, Cederbaum:2015fra, Rogatko:2024nzq}, Einstein-Maxwell-dilaton \cite{Yazadjiev:2015mta,Rogatko:2016mho}, Einstein-multiple-scalar \cite{Yazadjiev:2021nfr} models, for wormholes \cite{Koga:2020gqd,Yazadjiev:2017twg}, and in higher dimensions \cite{Jahns:2019lex,Cederbaum:2024zqd}. The condition of constancy of the lapse function that was used in first papers was weakened to constancy of lapse function in each spatial slice in Ref. \cite{Cederbaum:2019rbv}. An alternative perturbative approach to considering the uniqueness of photon spheres was suggested in Ref. \cite{Yoshino:2016kgi, Yoshino:2023mcs}. Gibbons and Warnick \cite{Gibbons:2016isj} discovered that photon surfaces may exist in less symmetric spacetimes. This was extended in Ref. \cite{Koga:2020akc}, where the photon surfaces were demonstrated to be possible if the spacetime admits a nontrivial Killing tensor. This might suggest the potential challenge of formulating a uniqueness theorem solely relying on the umbilicity without constancy of the lapse function. Similar connections in stationary spacetimes was established in the Refs. \cite{Kobialko:2021aqg,Kobialko:2022ozq}.

The framework of photon spheres has been extended to encompass \textit{massive particle surfaces}, which share analogous properties for timelike geodesics associated with massive particles interacting with black holes or other ultra-compact gravitating objects \cite{Kobialko:2022uzj,Bogush:2023ojz,Bogush:2024fqj}. Although the flow of massive particles is not directly observable from distant points (except perhaps for neutrinos, whose detection remains a considerable challenge), these surfaces can be indirectly observed through their proper radiation, which may become visible under certain conditions. Moreover, the significance of massive particle surfaces lies in their relevance to photons traversing through plasma environments that may surround black holes \cite{Kobialko:2023qzo}. In environments with inhomogeneous plasma, besides the gravitational deflection of light, electromagnetic refraction also plays a role \cite{Bisnovatyi-Kogan:2010flt,Wagner:2020ihx,Fathi:2021mjc,Kumar:2023wfp}, which can be integrated into a unified lensing theory.

In this paper, we adapt proofs from Refs. \cite{Yazadjiev:2015hda,Yazadjiev:2015jza,Yazadjiev:2015mta,Yazadjiev:2021nfr} suggested by Yazadjiev et. al. Within the framework of the same assumptions (there is only one connected photon sphere and the lapse function regularly foliates spacetime), we prove the uniqueness theorem for massive particle spheres in static vacuum asymptotically flat spacetimes. In contrast to previous uniqueness theorems, our result leads to the existence of an entire spacetime foliation sliced by a set of massive particle spheres spanning various energies.

\section{Massive particle sphere}

\setcounter{equation}{0}
\setcounter{subsection}{0}

We start with considering a four-dimensional static vacuum asymptotically flat spacetime $\mathcal{M}$ with given ADM mass $M>0$ and Levi-Civita connection $\LCM_\alpha$. In a static spacetime, there is a timelike Killing vector field $k^{\alpha}=\alpha m^\alpha$, where $\alpha>0$ is a lapse function and $m^\alpha$ is a future-directed timelike unit vector ($k^\alpha k_\alpha = -\alpha^2$). One can define spatial slices $\Sigma$ orthogonal to $k^\alpha$. The induced metric on $\Sigma$ is $\bar{g}_{\alpha\beta}=g_{\alpha\beta} + m_\alpha m_\beta$ which defines the corresponding Levi-Civita connection $\bar{\LCM}_\alpha$. Here and further bars will denote quantities associated with the slice $\Sigma$. Vacuum Einstein equations $R_{\alpha\beta}=0$ after dimensional reduction along $k^\alpha$ read \cite{Cederbaum:2014gva,Yazadjiev:2015hda}
\begin{align}\label{eq:Einstein}
\bar{R}{}_{\alpha\beta} =\alpha^{-1}  \bar{\LCM}_\alpha \bar{\LCM}_\beta\alpha, \quad \bar{\LCM}_\alpha \bar{\LCM}^\alpha\alpha =0.
\end{align}
Since we are interested in asymptotically flat spacetimes, the lapse function and the metric have the following asymptotics for $r\to\infty$ \cite{Cederbaum:2014gva,Yazadjiev:2015hda}:
\begin{align} \label{eq:flat}
\alpha=1-\frac{M}{r}+O\left(r^{-2}\right), \quad  g_{\alpha\beta}=\eta_{\alpha\beta} + O\left(r^{-1}\right),
\end{align}
where $r$ is a suitable radial coordinate and $\eta_{\alpha\beta}$ is a flat Minkowski four-dimensional metric. For vacuum spacetime with a time-like Killing vector $k^{\alpha}$, there is an alternative in determination of the ADM mass of the solution through the Komar integral \cite{Komar:1963svp}
\begin{align} \label{eq:Komar}
M = -\frac{1}{8\pi}\int_{\bar{S}}  \LCM^\alpha k^\beta dS_{\alpha\beta}, \quad dS_{\alpha\beta}=n_{[\alpha} m_{\beta]}d\bar{S},
\end{align}
where $\bar{S}$ is an arbitrary closed two-dimensional surface in $\Sigma$ with an outer normal vector $n^\alpha$ (vector $n^\alpha$ lies in the tangent space of slice $\Sigma$) and $d\bar{S}$ is a volume form associated with the induced metric on $\bar{S}$.

The definition of a massive particle surface for static vacuum spacetimes can be formulated as follows \cite{Kobialko:2022uzj}: 

\begin{definition} 
A massive particle surface is a timelike hypersurface $S$ of $\mathcal{M}$ such that, for every point $p\in S$ and every vector $v^\alpha|_p \in T_pS$ such that $v^\alpha  k_\alpha|_p=-\E$ and $v^\alpha v_\alpha|_p=-m^2$, there exists a geodesic $\gamma$ of $\mathcal{M}$ for a particle with mass $m$ and energy $\E$ such that $\dot{\gamma}^\alpha(0) =v^\alpha|_p$ and $\gamma\subset S$.
\end{definition} 

We omit the charge from the definition in Ref. \cite{Kobialko:2022uzj} since the solution is vacuum and there is no electromagnetic force acting on particles. Here, we focus on \textit{static massive particle surface}, which are additionally tangent to the timelike Killing vector $k^\alpha$. 

In other words, the definition states that any geodesic of a particle with mass $m$ and energy $\E$ initially tangent to the corresponding massive particle surface $S$ will remain tangent to $S$. For static massive particle surface $S$ with normal $n^\alpha$ the first and second fundamental forms read as \cite{Kobialko:2022uzj}
\begin{align}\label{eq:h_chi}
   h_{\alpha\beta}=g_{\alpha\beta}-n_\alpha n_\beta, \quad
    \chi_{\alpha\beta} = H\left(
          h_{\alpha\beta}
        + \frac{m^2}{ \E^2} k_{\alpha} k_{\beta}
    \right),
\end{align}
where $H$ is some scalar function on $S$ and $\LCS_\alpha$ is a Levi-Civita connection in $S$. Since the Killing vector field $k^\alpha$ is tangent to the hypersurface $S$ everywhere ($k^\alpha n_\alpha =0$), the following Lie derivatives are equal to zero (see App. \ref{AppA}):
\begin{align}\label{eq:Lie}
    \Lie_k n_\alpha =0, \quad
    \Lie_k h_{\alpha\beta}=0, \quad
    \Lie_k  \chi_{\alpha\beta}=0, \quad
    \Lie_k H = 0.
\end{align}
If $\bar{S}$ is a spatial section of a surface $S$ sliced by $\Sigma$, from general geometric considerations we have \cite{Yoshino:2017gqv}:
\begin{align}\label{eq:chi_bar}
   \chi_{\alpha\beta} = \bar{\chi}_{\alpha\beta} - m_\alpha m_\beta \cdot n^\alpha  \LCM_\alpha \ln \alpha, \quad
   h_{\alpha\beta} = \bar{h}_{\alpha\beta} - m_\alpha m_\beta.
\end{align}
Comparing (\ref{eq:h_chi}) and (\ref{eq:chi_bar}), we find 
\begin{align}
    \bar{\chi}_{\alpha\beta} =
      H \bar{h}_{\alpha\beta}
    +  \alpha^{-2} k_\alpha k_\beta \left(
        n^\alpha  \LCM_\alpha \ln \alpha
        - H \left(1 - \frac{\alpha^{2} m^2}{\E^2}\right)
    \right),
\end{align}
and since $\bar{\chi}_{\alpha\beta}$ is tangent to the spatial section and $k_\alpha$ is orthogonal to it, we find the following expressions:
\begin{align}\label{eq:chi_s}
\bar{\chi}_{\alpha\beta}  &=H \bar{h}_{\alpha\beta}, \quad
  n^\alpha \LCM_\alpha \ln \alpha = H\left(1
        -\frac{\alpha^2m^2}{\E^2} 
    \right).
\end{align}
The spatial section of static massive particle hypersurface is a totally umbilical surface with a spatial mean curvature $H=\frac{1}{2}\bar{\chi}_{\alpha}{}^\alpha$. However, unlike the photon sphere \cite{Claudel:2000yi}, the principal curvature in the time direction is different from the spatial ones. In what follows we also assume that the spatial section is connected, compact and closed.  

Equations (\ref{eq:Lie}), (\ref{eq:chi_s}) and Refs. \cite{Cederbaum:2014gva,Cederbaum:2015fra,Yazadjiev:2017twg,Yazadjiev:2015hda,Yazadjiev:2015jza,Yazadjiev:2015mta,Yazadjiev:2021nfr} inspire us to introduce two important definitions -- a massive particle \textit{sphere} and a \textit{non-extremal} massive particle surface.

\begin{definition}
The massive particle surface $S$ is a massive particle sphere if and only if
$\LCS_\alpha \alpha =0$ on $S$.
\end{definition}

\begin{definition}
The massive particle surface $S$ is a non-extremal massive particle surface if and only if $m^2\alpha^2 /\E^2<1$ on $S$.
\end{definition}

Further we will be interested in \textit{non-extreme massive particle spheres} only. The latter has a number of important geometric properties. First, it has a constant spatial mean curvature, i.e. $\LCS_\alpha H = 0$. Indeed, consider the Codazzi equation \cite{Chen} in spacetime
\begin{align}
    0
    =
    & n^\rho h^\sigma{}_\alpha \RM_\rho{}_{\sigma}
    =
    \LCS_\beta \SFS_{\alpha}{}^\beta-\LCS_\alpha \SFS_{\beta}{}^\beta=
    \\\nonumber = &
      \frac{m^2}{ \E^2} k_{\alpha} \Lie_k H 
    + \frac{Hm^2}{ \E^2}\LCS_\beta\left(
            k_{\alpha} k^{\beta}
    \right)
    + \LCS_\alpha H
    -\left(
          3 - \frac{m^2\alpha^2}{ \E^2} 
    \right) \LCS_\alpha H
    \\\nonumber = &
    -\left(
          2
        - \frac{m^2\alpha^2}{ \E^2} 
    \right)\LCS_\alpha H.
\end{align}
Here, we used Eq. (\ref{eq:h_chi}) to rewrite the second fundamental form $\chi_{\alpha\beta}$. Then, we used Eq. (\ref{eq:Lie}) and 
\begin{align}
\LCS_\beta\left(
            k_{\alpha} k^{\beta}
    \right)=\frac{1}{2}\LCS_\alpha
            \alpha^2
    +
            k_{\alpha}\cdot \LCS_\beta k^{\beta}=0,
\end{align}
to get rid of the first two terms in the second line. Also, we used the condition $\LCS_\alpha \alpha =0$ and Killing equations $\LCS_{(\alpha} k_{\beta)}=0$. Since we consider a non-extremal sphere $m^2 \alpha^2 /\E^2<1$, the term in brackets $(2- m^2\alpha^2 / \E^2)$ is non-zero, and we get $\LCS_\alpha H=0$. 

Since we have proven that $H$ is constant at the sphere, this allows us to obtain useful geometric identities. First, consider the Komar integral (\ref{eq:Komar}) over the spatial sections $\bar{S}$ of a massive particle sphere:
\begin{align} \label{eq:Mk}
M &= -\frac{1}{8\pi}\int_{\bar{S}}  \LCM^\alpha k^\beta dS_{\alpha\beta}=  -\frac{1}{8\pi}\int_{\bar{S}}  \LCM^\alpha k^\beta n_{[\alpha} m_{\beta]}d\bar{S} \nonumber\\&=
 \frac{1}{4\pi}\int_{\bar{S}}n^\alpha\LCM_\alpha \alpha d\bar{S}=\frac{1}{4\pi}\int_{\bar{S}} \alpha H\left(1-\frac{\alpha^2m^2}{\E^2} \right)d\bar{S}.
\end{align}
Since the integrand expression is constant, there is an algebraic relation between the mass $M$ and the spatial section area of the massive particle sphere $A_S$:
\begin{align} \label{eq:ADM}
4 \pi M  = \alpha H\left(1-\frac{\alpha^2m^2}{\E^2} \right) A_S, \qquad
A_S=\int_{\bar{S}}d\bar{S}.
\end{align}
Particularly, sphere $\bar{S}$ has a positive constant mean curvature $H>0$ if a physical assumption of positive mass $M$ is taken into consideration. By virtue of Eq. (\ref{eq:chi_s}), this means that on the sphere $n^\alpha \LCM_\alpha \alpha >0$, i.e., the norm of spatial gradient $\bar{\LCM}_\alpha \alpha$ does not vanish anywhere on $\bar{S}$.

Consider an outer space region $\Sigma_{\text{ext}}$ outside the massive particle sphere $\bar{S}$ or equivalently a spacetime region $\mathcal{M}_{\text{ext}}$ outside $S$. In this case, the massive particle sphere is an inner boundary $\partial \mathcal{M}_{\text{ext}}$. Similarly to Ref. \cite{Yazadjiev:2015hda}, we introduce an additional assumption that $\alpha={const}$ regularly foliate the manifold $\mathcal{M}_{\text{ext}}$. It is worth noting that the condition for the existence of a regular foliation is technical and, in principle, open to relaxation, as discussed in Ref. \cite{Cederbaum:2015fra}. 

By definition, the function $\alpha$ is constant at the massive particle sphere. As we will show, the massive particle sphere $\bar{S}$ has a topology of a sphere. Given the regularity of the foliation, any slice in the outer region $\Sigma_{\text{ext}}$ is a topological sphere as well. Equations of motion (\ref{eq:Einstein}) necessitate $\alpha$ to be a harmonic function, while the boundary conditions at asymptotics dictate that $\alpha$ must approach 1 as it tends to infinity. Following the maximum principle for the harmonic functions, $\alpha$ monotonically increases to 1 moving from the sphere $\bar{S}$ to infinity along the flow of slices, i.e., $0<\alpha<1$.

The second key identity can be obtained from the Gauss-Bonnet theorem. The trace of the Gauss equations gives an expression for the scalar curvature $\bar{\RS}$ of the spatial section $\bar{S}$ the following (see Eq. (C13) in Ref. \cite{Yoshino:2017gqv}, keeping in mind that $R_{\alpha\beta}=0$, $\LCS_\alpha \alpha =0$ and (\ref{eq:chi_s}))
\begin{align} \label{eq:Gauss}
\bar{\RS}&=\bar{\SFS}^2 - \bar{\SFS}_{\alpha\beta}\bar{\SFS}^{\alpha\beta} +2 \bar{\SFS}_\beta{}^\beta n^\alpha \LCM_\alpha \ln \alpha  
=
4\left(\frac{3}{2}-\frac{\alpha^2m^2}{\E^2} \right)H^2.
\end{align}
Therefore, the non-extremal ($m^2\alpha^2 /\E^2<1$) sphere $\bar{S}$ has a constant and positive scalar curvature $\bar{\RS}>0$, representing a round
sphere \cite{Yazadjiev:2015hda}. Then, integrating (\ref{eq:Gauss}) over $\bar{S}$ and applying Gauss-Bonnet theorem $\int_{\bar{S}} \bar{\RS}d\bar{S} = 8\pi$, we find the second useful identity:
\begin{equation} \label{eq:gauss_bonnet}
2\pi=\left(\frac{3}{2}-\frac{\alpha^2m^2}{\E^2} \right)H^2 A_S.
\end{equation}
Dividing the equation (\ref{eq:gauss_bonnet}) by (\ref{eq:ADM}), the following algebraic connection between the mean curvature $H$ and the lapse function $\alpha$ on $S$ can be found:
\begin{equation} \label{eq:H}
H=\frac{\alpha}{M}\cdot\frac{1-\alpha^2 m^2/\E^2}{3 -2 \alpha ^2 m^2/\E^2}.
\end{equation}

\section{Uniqueness theorem}

Having completed all the preparations, we are ready to formulate and prove the main result of this article.

\begin{theorem}
Let $\mathcal{M}_{\text{ext}}$ be a four-dimensional static and asymptotically flat spacetime with given ADM mass $M>0$, satisfying the vacuum Einstein equations $R_{\alpha\beta}=0$ and possessing a non-extremal massive particle sphere as an inner boundary of $\mathcal{M}_{\text{ext}}$. Assume that the lapse function $\alpha$ regularly foliates $\mathcal{M}_{\text{ext}}$. Then, $\mathcal{M}_{\text{ext}}$ is an isometric to the Schwarzschild spacetime with mass $M$, and the area radius $r_S=\sqrt{A_S/4\pi}$ of the massive particle sphere satisfies the equation ${\E^2/m^2=(r_S-2M)^2/(r^2_S-3Mr_S)}$.
\end{theorem}

\begin{proof}
The proof is based on a modification of the proof presented in Refs. \cite{Yazadjiev:2015hda,Yazadjiev:2015jza,Yazadjiev:2015mta,Yazadjiev:2021nfr} for the case of photon spheres. The main problem is to prove the spherical symmetry of the spacetime $\mathcal{M}_{\text{ext}}$. First, let us perform a Weyl transformation $\tilde{g}_{\alpha\beta}=\alpha^{2}\bar{g}_{\alpha\beta}$. In this case, Eq. (\ref{eq:Einstein}) turns into  
\begin{align} \label{eq:Einstein_Weyl}
\tilde{R}{}_{\alpha\beta} =2\tilde{\LCM}_\alpha \ln\alpha \cdot\tilde{\LCM}_\beta\ln\alpha, \quad \tilde{\LCM}_\alpha \tilde{\LCM}^\alpha\ln\alpha=0,
\end{align}
where $\tilde{\LCM}$ and $\tilde{R}{}_{\alpha\beta}$ are the Levi-Civita connection and the Ricci tensor for $\tilde{g}_{\alpha\beta}$. Our goal is to show that metric $\tilde{g}_{\alpha\beta}$ is conformally flat. For this purpose, one can use the Cotton tensor \cite{Garcia:2003bw} over a 3-dimensional Riemannian manifold which is defined by
\begin{align}
\tilde{R}_{\alpha\beta\gamma}= \tilde{\LCM}_{[\alpha} \tilde{R}_{\beta]\gamma} - \frac{1}{4} \tilde{\LCM}_{[\alpha} \tilde{R}  \tilde{g}_{\beta]\gamma}. 
\end{align}
Using Eq. (\ref{eq:Einstein_Weyl}), the following divergences can be obtained \cite{Yazadjiev:2015hda}:
\begin{subequations} 
\begin{align}
\tilde{\LCM}_{\alpha} \left(\Omega^{-1}\tilde{\LCM}^{\alpha}\omega\right)&=\frac{1}{16}\omega^{-7} \Omega^3 \tilde{R}_{\alpha\beta\gamma}\tilde{R}^{\alpha\beta\gamma}, \\
\tilde{\LCM}_{\alpha} \left(\Omega^{-1}\left(U\tilde{\LCM}^{\alpha}\omega-\omega \tilde{\LCM}^{\alpha} U\right)\right)&=\frac{1}{16}H\omega^{-7} \Omega^3 \tilde{R}_{\alpha\beta\gamma}\tilde{R}^{\alpha\beta\gamma},
\end{align}
\end{subequations} 
where
\begin{align}
\omega=\left(\tilde{\LCM}_{\alpha}U\tilde{\LCM}^{\alpha}U\right)^{1/4}, \quad U = \frac{1-\alpha^2}{1+\alpha^2}, \quad \Omega=\frac{4\alpha^2}{(1+\alpha^2)^2}.
\end{align}
Since we have shown that $0<\alpha<1$, we also have $0<U<1$. Which after integration over the entire spatial slice $\Sigma$ lead to the following two inequalities (equality if and only if $\tilde{R}_{\alpha\beta\gamma}=0$) 
\begin{align}  \label{eq:neq_old}
\int_{\Sigma}\tilde{\LCM}_{\alpha} \left(\Omega^{-1}\tilde{\LCM}^{\alpha}\omega\right)d\tilde{\Sigma}\geq\int_{\Sigma} \tilde{\LCM}_{\alpha} \left(\Omega^{-1}\left(U\tilde{\LCM}^{\alpha}\omega-\omega \tilde{\LCM}^{\alpha} U\right)\right) d\tilde{\Sigma}\geq 0,
\end{align}
where $d\tilde{\Sigma}$ volume form associated with metric $\tilde{g}_{\alpha\beta}$.
Let us now apply Stokes’ theorem to them using massive particle sphere $\bar{S}$ and asymptotic sphere $\bar{S}_\infty$ as boundary surfaces: 
\begin{align} \label{eq:neq}
\int_{\bar{S}_\infty-\bar{S}} \alpha\Omega^{-1}n^{\alpha}\bar{\LCM}_{\alpha}\omega d\bar{S}\geq\int_{\bar{S}_\infty-\bar{S}} \alpha \Omega^{-1}\left(Un^{\alpha}\bar{\LCM}_{\alpha}\omega-\omega n^{\alpha}\bar{\LCM}_{\alpha} U\right) d\bar{S}\geq 0,
\end{align}
where we used $d\tilde{S}=\alpha^2d\bar{S}$ and $\tilde{n}^{\alpha}=\alpha^{-1}n^{\alpha}$, and $-\bar{S}$ means that the orientation of the normal to the inner boundary is opposite to the foliation. Given asymptotics  (\ref{eq:flat}), each surface term reads (see App. \ref{AppB} for some details)
\begin{subequations}\label{eq:int}
\begin{align}
&
\int_{\bar{S}_\infty} \alpha\Omega^{-1}n^{\alpha}\bar{\LCM}_{\alpha}\omega d\bar{S}=-4 \pi \sqrt{M},
\\&
\int_{\bar{S}_\infty} \alpha \Omega^{-1}\left(Un^{\alpha}\bar{\LCM}_{\alpha}\omega-\omega n^{\alpha}\bar{\LCM}_{\alpha} U\right) d\bar{S}=0,
\\&
\int_{-\bar{S}} \alpha\Omega^{-1}n^{\alpha}\bar{\LCM}_{\alpha}\omega d\bar{S}=\frac{A_S}{2}\sqrt{\frac{H^3}{\alpha} \left(1-\frac{\alpha ^2 m^2}{\E^2}\right)} \left(\left(1+3 \alpha ^2\right)-\frac{2 \alpha ^4 m^2}{\E^2}\right), 
\\&
\int_{-\bar{S}} \alpha \Omega^{-1}\left(Un^{\alpha}\bar{\LCM}_{\alpha}\omega-\omega n^{\alpha}\bar{\LCM}_{\alpha} U\right) d\bar{S}=\frac{A_S}{2}\sqrt{\frac{H^3}{\alpha} \left(1-\frac{\alpha ^2 m^2}{\E^2}\right)} \left(\left(1-3 \alpha ^2\right)+\frac{2 \alpha ^4 m^2}{\E^2}\right).
\end{align} 
\end{subequations}
Then, the right inequality in (\ref{eq:neq}) immediately results in
\begin{align} \label{eq:inequality_1}
3 \alpha ^2-1-\frac{2\alpha^4 m^2}{\E^2}\leq 0.
\end{align}
The left inequality in (\ref{eq:neq}) gives 
\begin{align}
4 \pi \sqrt{M} -A_S\sqrt{\alpha^3H^3 \left(1-\frac{\alpha^2 m^2}{\E^2}\right)} \left(3 -\frac{2 \alpha^2 m^2}{\E^2}\right)\leq 0,
\end{align}
which can be transformed using Eqs. (\ref{eq:gauss_bonnet}) and (\ref{eq:H}) into the following expression
\begin{align} \label{eq:inequality_2}
3 \alpha ^2-1-\frac{2\alpha^4 m^2}{\E^2}\geq 0,
\end{align}
where we took into account non-extrimality $m^2\alpha^2 /\E^2<1$. The inequalities (\ref{eq:inequality_1}) and (\ref{eq:inequality_2}) are compatible if and only if they degenerate into equalities
\begin{align} \label{eq:equality}
3 \alpha ^2-1-\frac{2\alpha^4 m^2}{\E^2}=0.
\end{align}
On the other hand, inequalities can degenerate into equalities if and only if the Cotton tensor vanishes $\tilde{R}_{\alpha\beta\gamma}=0$. For a three-dimensional Riemannian manifold, this is a necessary and sufficient condition for the metric $\tilde{g}_{\alpha\beta}$ to be conformally flat \cite{Garcia:2003bw}. Hence, the metric $\bar{g}_{\alpha\beta}$ is also conformally flat and $\bar{R}_{\alpha\beta\gamma}=0$. In particular, we have the identity \cite{Yazadjiev:2015hda}
\begin{align}
0=\bar{R}_{\alpha\beta\gamma}\bar{R}^{\alpha\beta\gamma}=\frac{8}{\alpha^4\varphi^4}\left(\left({}^\alpha\bar{\SFS}_{\alpha\beta}-\frac{{}^\alpha\bar{\SFS}}{2}\cdot{}^\alpha\bar{h}_{\alpha\beta}\right)\left({}^\alpha\bar{\SFS}^{\alpha\beta}-\frac{{}^\alpha\bar{\SFS}}{2}\cdot{}^\alpha\bar{h}^{\alpha\beta}\right) + \frac{1}{2\varphi^2} {}^\alpha\bar{h}^{\alpha\beta} \partial_\alpha\varphi \partial_\beta\varphi \right), 
\end{align}
where ${}^\alpha\bar{\SFS}_{\alpha\beta}$, ${}^\alpha\bar{h}_{\alpha\beta}$ and  $\varphi^{-1}=n^\alpha\LCM_{\alpha} \alpha$ are the induced metric, the second fundamental forms and the lapse function of slices $\alpha=\text{const}$ respectively, and the trace is denoted as ${}^\alpha\bar{\SFS}={}^\alpha\bar{\SFS}_{\alpha}{}^\alpha$. Since the induced metric possesses the Euclidean signature, each square in the brackets is equal to zero, yielding the following expressions:
\begin{align} \label{eq:hhh}
{}^\alpha\bar{\SFS}_{\alpha\beta}=\frac{{}^\alpha\bar{\SFS}}{2}\cdot{}^\alpha\bar{h}_{\alpha\beta}, \quad \bar{\LCS}_\alpha\varphi=0.
\end{align}

Thus, all slices are totally umbilic and the lapse function is constant on them. As in the case of photon spheres \cite{Yazadjiev:2015hda}, this implies that all slices of the foliation $\alpha=\text{const}$ have constant mean and scalar curvatures, i.e. slices are round spheres. As a result, the entire spacetime $\mathcal{M}_{\text{ext}}$ is spherically symmetric and therefore isometric to the Schwarzschild vacuum asymptotically flat spacetime (due to Birkhoff’s theorem). In particular, resolving Eqs. (\ref{eq:ADM}), (\ref{eq:gauss_bonnet}), (\ref{eq:equality}) and introducing the area radius
\begin{align}
r_S=\sqrt{\frac{A_S}{4\pi}}, 
\end{align}
we get standard expressions for the radius of the massive particles sphere \cite{Kobialko:2022uzj} and lapse functions in the Schwarzschild spacetime ($m\neq0$)
\begin{align}\label{eq:mps}
\frac{\E^2}{m^2}=\frac{(r_S-2M)^2}{r_S(r_S-3M)}, \quad \alpha^2=1-\frac{2M}{r_S}, \quad H^2=\frac{1}{r^2_S}\left(1-\frac{2M}{r_S}\right).
\end{align}
This completes the proof of the theorem. 
\end{proof}

On the one hand, substitution of Eq. (\ref{eq:mps}) into the condition of nonextremality $m^2\alpha^2 /\E^2<1$ results in ${M/(r_S-2M)> 0}$, i.e. holds outside the horizon. On the other hand, massive particle spheres exists for $r_S > 3M$, otherwise $\E^2/m^2$ is negative. There is a photon sphere at $r_S = 3M$, so massive particle spheres are located outside the photon sphere, which is a physically reasonable.

We also emphasize the need to have only one massive particle sphere  to prove the theorem. However, the result of the theorem suggests that the entire spacetime $\mathcal{M}_{\text{ext}}$ is sliced by the massive particle spheres, each with distinct energy. Indeed, by virtue of (\ref{eq:hhh}) all spatial slices $\alpha={\text{const}}$ are totally umbilic and have a constant mean curvature and lapse function $\varphi^{-1}=n^\alpha\LCM_{\alpha} \alpha$. 
These slices represent a massive particle sphere when we additionally demand only that Eq. (\ref{eq:chi_s}) admits a real solution for $\E$, as it will automatically remain constant on the slice. From our previous discussion, it is clear that such a solution will exist for all slices at $r_S > 3M$.  Future inquiries may find it intriguing to explore the flows of massive particle surfaces, parameterized by particle energy, rather than adhering to a regular foliation $\alpha={\text{const}}$. Such a shift could potentially weaken several technical assumptions of the theorem.

\section{Conclusions}

In this paper, we have established (under some technical assumptions) that a four-dimensional static vacuum asymptotically flat spacetime admitting a massive particle sphere is isometric to the Schwarzschild spacetime. This broadens the scope of uniqueness theorems applicable to static vacuum asymptotically flat spacetimes containing regular event horizons or photon spheres, now encompassing a more general case of massive particles. Moreover, this result allows further generalization to other theories like Einstein-scalar, Einstein-Maxwell, and Einstein-Maxwell-dilaton-axion theories and others.

It is worth noting that the new theorem offers the possibility of extending the uniqueness theorem to encompass strongly naked singularities \cite{Virbhadra:2002ju}, wherein neither a photon sphere nor a horizon exists, but a massive particle sphere is present. For instance, in a superextreme electro-vacuum spacetime, the massive particle sphere exists in a broader range of parameters compared to the photon sphere and, particularly, can be detected in close proximity to a strongly naked singularity \cite{Kobialko:2022uzj}.

While the assumption of a regular foliation is just technical \cite{Yazadjiev:2015hda,Cederbaum:2015fra}, the constancy of the lapse function and the static nature of the sphere plays a key role in the proof of the uniqueness theorem. Recent work has explored the notion of equipotential surfaces as a potential dynamic alternative to static sphere \cite{Cederbaum:2019rbv}. However, whether solely relying on the concept of a massive particle surface is sufficient for the uniqueness theorems, remains uncertain. Unlike unique photon surfaces, massive particle surfaces form entire flows (for varying energies) that extend to infinity, passing in asymptotic spheres. Analysis of such surface flows can provide additional information and advances in this area of research.

In addition, there is considerable interest in the prospect of extending the result to stationary spacetime, where there is a suitable geometric definition of the surfaces of massive particles \cite{Bogush:2023ojz}. Such generalizations could expand understanding of the role of massive particle surfaces and hidden symmetries in the discussion of uniqueness.

\begin{acknowledgments}
KK and DG acknowledge the support of the Russian Science Foundation under Contract No. 23-22-00424.
\end{acknowledgments}

\appendix

\section{Proposition}
\label{AppA}

\begin{proposition} \label{prop:killing}
Let the Killing vector field $k^\alpha$ be everywhere tangent ($k^\alpha n_\alpha =0$) to the hypersurface $S$, then 
\begin{align}
\Lie_k n_\alpha =0, \quad \Lie_k h_{\alpha\beta}=0, \quad \Lie_k  \chi_{\alpha\beta}=0. 
\end{align}
\end{proposition}

\begin{proof}
Calculate the Lie derivative of the normal covector
\begin{align}
\Lie_k n_\alpha = k^\beta \LCM_\beta n_\alpha + n_\beta  \LCM_\alpha  k^\beta.
\end{align}
The projection of this equation onto the normal $n^\alpha$ reads
\begin{align}
n^\alpha \Lie_k n_\alpha =n^\alpha  k^\beta \LCM_\beta n_\alpha + n^\alpha  n_\beta  \LCM_\alpha  k^\beta=\frac{1}{2} k^\beta \LCM_\beta (n^\alpha n_\alpha ) +  n^\alpha  n^\beta \LCM_\alpha  k_\beta=0,
\end{align}
by virtue of the Killing equations $\LCM_{(\alpha} k_{\beta)}=0$ and normalization of the normal vector on the surface $n^\alpha n_\alpha = 1$. The tangent projection reads
\begin{align}
h^\alpha_\gamma \Lie_k n_\alpha = k^\beta h^\alpha_\gamma \LCM_\beta n_\alpha + n_\beta  h^\alpha_\gamma \LCM_\alpha  k^\beta=k^\beta h^\alpha_\gamma \LCM_\beta n_\alpha -k^\beta h^\alpha_\gamma \LCM_\alpha n_\beta=-k^\sigma  h^\alpha_\gamma h^\beta_\sigma \LCM_{[\alpha} n_{\beta]}=0,
\end{align}
where we used the relation
\begin{align}
0=h^\alpha_\gamma \LCM_\alpha (k_\beta n^\beta)=n^\beta h^\alpha_\gamma \LCM_\alpha k_\beta +k_\beta h^\alpha_\gamma \LCM_\alpha n^\beta,
\end{align}
and the involutive property $h^\alpha_\gamma h^\beta_\sigma \LCM_{[\alpha} n_{\beta]}=0$. Thus, the expression $\Lie_k n_\alpha =0$ is proved. We also derive the following straightforward yet valuable corollaries
\begin{subequations}
\begin{align}
\Lie_k n^\alpha &=g^{\alpha\beta}\Lie_k n_\beta =0, \\
\Lie_k h_{\alpha\beta}&=\Lie_k (g_{\alpha\beta}-n_\alpha n_\beta)=- n_{(\alpha} \Lie_k n_{\beta)}=0, \\
\Lie_k  \chi_{\alpha\beta}& =  \frac{1}{2}\Lie_k \Lie_n h_{\alpha\beta}= \frac{1}{2}\Lie_n\Lie_k  h_{\alpha\beta}+\frac{1}{2}\Lie_{\Lie_k n} h_{\alpha\beta}=0.
\end{align}
\end{subequations}
Calculating the Lie derivative of equations (\ref{eq:h_chi}) we immediately find that $H$ is also static:
\begin{align}
\Lie_k H =0. 
\end{align}
\end{proof}

\section{Calculations} 
\label{AppB}
Here, we give explicit expressions for $\omega$ and its derivative along $n^\alpha$. First, from $\LCS_\alpha \alpha =0$ follows the expression
\begin{align} \label{eq:app_1}
\omega=(-\alpha^{-1} n^\alpha\bar{\LCM}_\alpha U)^{1/2}=\left(\frac{4}{(1+\alpha)^2}n^\alpha\LCM_\alpha \alpha\right)^{1/2}, 
\end{align}
where relations $\tilde{\LCM}_{\alpha}U=\bar{\LCM}_{\alpha}U=n_\alpha n^\beta \bar{\LCM}_{\beta} U$ and $\tilde{g}^{\alpha\beta}=\alpha^{-2}\bar{g}^{\alpha\beta}$ are used. Then, the remaining derivatives read
\begin{subequations}
\begin{align}\label{eq:app_2_a}
n^\beta \bar{\LCM}_\beta U &=- \frac{4\alpha}{(1+\alpha^2)^2}n^\beta \bar{\LCM}_\beta\alpha, \\ \label{eq:app_2_b}
n^\beta \bar{\LCM}_\beta\omega&=\frac{1}{2} (-\alpha^{-1} n^\alpha\bar{\LCM}_\alpha U)^{-1/2} \left(-\frac{16 \alpha  (n^\alpha\LCM_\alpha \alpha)^2}{\left(\alpha ^2+1\right)^3}-\frac{8 H}{\left(\alpha ^2+1\right)^2}n^\alpha\LCM_\alpha \alpha\right),
\end{align}
\end{subequations}
where we used identity $n^\alpha \bar{\LCM}_\alpha (n^\beta \bar{\LCM}_\beta \alpha)=-2Hn^\gamma \LCM_\gamma\alpha$.
Using equations (\ref{eq:app_1}), (\ref{eq:app_2_a}), (\ref{eq:app_2_b}) and identity (\ref{eq:chi_s}), some algebraic calculations lead us to the result (\ref{eq:int}).

\bibliography{main}

\end{document}